\documentclass[12pt]{article}

\usepackage{amsmath}
\usepackage{amsthm}
\usepackage{booktabs}
\usepackage{setspace}
\usepackage[margin=1in]{geometry}
\usepackage{natbib}
\usepackage{longtable}
\usepackage{titlesec}
\titleformat*{\section}{\large\bfseries}
\titleformat*{\subsection}{\large\bfseries}
\titleformat*{\subsubsection}{\normalsize\bfseries}

\newtheorem{definition}{Definition}
\newtheorem{thm}{Theorem}
\numberwithin{table}{section}
\numberwithin{figure}{section}

\newcommand{\T}{\mathrm{T}}

\title{Summary of effect aliasing structure (SEAS): \\ new descriptive statistics for factorial and supersaturated designs}

\author{Frederick Kin Hing Phoa$^\star$\\[.5ex]
	Yi-Hua Liao$^\star$\\[.5ex]
     David C. Woods$^\dagger$\\[.5ex]
    Shyh-Kae Chou$^\star$\\[1ex]
        $^\star$Institute of Statistical Science, Academia Sinica, Taiwan\\[.5ex]
    $^\dagger$Southampton Statistical Sciences Research Institute,\\ University of Southampton, UK}

\begin{document}
\thispagestyle{empty}
\maketitle

In the assessment and selection of supersaturated designs, the aliasing structure of interaction effects is usually ignored by traditional criteria such as $E(s^2)$-optimality. We introduce the Summary of Effect Aliasing Structure (SEAS) for assessing the aliasing structure of supersaturated designs, and other non-regular fractional factorial designs, that takes account of interaction terms and provides more detail than usual summaries such as (generalized) resolution and wordlength patterns. The new summary consists of three criteria, abbreviated as MAP: (1) Maximum dependency aliasing pattern; (2) Average square aliasing pattern; and (3) Pairwise dependency ratio. These criteria provided insight when traditional criteria fail to differentiate between designs. We theoretically study the relationship between the MAP criteria and traditional quantities, and demonstrate the use of SEAS for comparing some example supersaturated designs, including designs suggested in the literature. We also propose a variant of SEAS to measure the aliasing structure for individual columns of a design, and use it to choose assignments of factors to columns for an $E(s^2)$-optimal design. \\[1ex]


\newpage


\section{Introduction}
Factor screening via designed experiments and statistical modelling is an essential step in the successful analysis of many large-scale scientific, technical and industrial processes. Nonregular fractional factorial designs, including Plackett-Burman designs and other orthogonal main effects plans \citep{X2015}, have long proved popular for experimentation on systems where factor main effects are expected to dominate the response surface. When factor sparsity \citep{BM1986} is expected to hold, a supersaturated design (SSD), a nonregular fractional factorial design with fewer runs than factors first introduced in the discussion of the paper by \citet{S1959}, can prove a cost-effective and efficient choice. There are increasing examples of the successful application of SSDs, particularly in the chemical industries (for example, \citealp{DV2008}). However, results from the application of SSDs can sometimes be inconclusive due to the bias and variance inflation caused by the partial aliasing between main effects, and between main effects and interactions. Such issues also afflict the analysis of orthogonal main effects plans when interactions are non-neglible.    

There are two common strategies to improve the analysis of nonregular designs, and supersaturated designs in particular. The first approach is to attempt to overcome the necessary inadequacies of the design brought about by the restriction of the runsize through the application of more sophisticated statistical modelling methods. The most common methods are stepwise-type regressions (for example, \citealp{L1993}, \citealp{P2014}), Bayesian variable selection (for example, \citealp{CHW1997}, \citealp{BFL2002}) and penalized least squares (shrinkage) methods (for example, \citealp{LL2003}, \citealp{PPX2009}). Mostly, these methods work via the application of factor, or effect, sparsity, for example through the addition of prior information or the shrinkage of small estimated effects to zero. The effectiveness of these methods therefore depends crucially on the validity of the effect sparsity principle; see \citet{MW2010} and \citet{DWDLV2014} for assessments and reviews. 

The second approach is to improve the properties of the design being used for the screening experiment. For general nonregular designs, performance is often summarized via generalized resolution \citep{DT1999} or the generalized word length pattern (\citealp{TD1999}); see Section~\ref{sec:theory}. Still the most common criterion used to select SSDs is $E(s^2)$-optimality \citep{BC1962}, which measures the nonorthogonality between pairs of factor columns in the design. For an $n$-run design with $m$ factors, an $E(s^2)$-optimal design minimizes 
\begin{equation*}\label{eq:es2}
E(s^2) = \sum_{i<j}s^2_{ij}/\binom{m}{2}.
\end{equation*}
where $s_{ij}$ is the vector inner-product between columns $i$ and $j$ of the $n\times m$ design matrix $\mathbf{X}$. More recently, \citet{MW2010} and \citet{JM2014} have generalized this criterion to also include inner-products between the factor columns and the constant (intercept) column.

Following \citet{BC1962}, methods of constructing optimal or efficient designs under the $E(s^2)$ criterion were take up by \citet{L1993}, who used half-fractions of Hadamard matrices, and \citet{W1993}, who added partially aliasing interaction columns to orthogonal arrays. More recently, algorithmic approaches have been developed, including a nature-inspired metaheuristic method, Swarm Intelligence Based (SIB), by \citet{PCWW2016}.

Other notable criteria for the selection of SSDs include: the $D_f$-criterion \citep{W1993}, which maximizes a sum of determinants of sub-matrices of $\mathbf{X}^{\T}\mathbf{X}$ corresponding to potential subsets of active factors; the $B$-criterion \citep{DL1994} which minimizes a measure of non-orthogonality between columns of the design obtained through regressions of one factor column on the other factor columns; and Bayesian $D$-optimality \citep{JLN2008} which minimizes the determinant of the posterior variance-covariance matrix for the main effects obtained using a sparsity-inducing prior distribution.

In general, the application of these design selection criteria has ignored the impact of interactions on the performance of the design and the subsequent modelling. The criteria focus, directly or indirectly, on correlations between pairs of factor columns in the design or, equivalently, on partial aliasing between main effects. Significant and substantive interaction effects can severely bias estimated main effects, resulting in misleading conclusions being drawn from the experiment. Equally, for many experiments performed using nonregular designs, interactions can be a major focus. This is particularly the case in chemistry experiments, when interactions between chemical constituents are usually anticipated \citep{PWX2009}. Traditional criteria often fail to assess the performance of SSDs for the estimation of interaction effects.   

In this paper, we propose and demonstrate a new descriptive summary for the aliasing structure of a SSD, or a nonregular factorial design in general. We call it the Summary of Effect Aliasing Structure (SEAS). In Section~\ref{sec:seas}, we introduce the components of SEAS and provide intuitive interpretations. In Section~\ref{sec:theory}, we state some theoretic results on the SEAS that provide connections to the traditional quantities of design resolution, word length pattern and $E(s^2)$. In Section~\ref{sec:comp}, we demonstrate the application of SEAS on some example SSDs, including some designs suggested in the literature. In Section~\ref{sec:col}, we propose a variant of the SEAS, called Effect-SEAS, to describe the partial aliasing of individual factorial effects with other effects in the design, and to quantify the impact of assigning factors to different columns of the design. A discussion and summary are provided in Section~\ref{sec:disc}, and tables of SEAS assessment of the designs from Sections~\ref{sec:comp} and~\ref{sec:col} are provided in the Appendices.

\section{Summary of effect aliasing structure}\label{sec:seas}
In this section we define our new summaries of effect aliasing. We first define the aliasing index, following \citet{CLY2004}. A $n$-run design in $m$ factiors is given by an $n\times m$ matrix $\mathbf{X}$, with $j$th column $\mathbf{x}_j$ a vector with entries either $-1$ or $+1$ for all $j=1,\dots,m$. Let $S= \left \{ \mathbf{c}_1,\ldots,\mathbf{c}_k \right \}$ be a subset of $k$ columns of $\mathbf{X}$, where $\mathbf{c}_i = \mathbf{x}_j$ for some $j\in\{1,\ldots.m\}$ and $\mathbf{c}_u\ne\mathbf{c}_v$ for all $u,v=1,\ldots,k$. Then the {\it aliasing index} of these $k$ columns in $S$, denoted as $\rho_k(S)$, is defined as
\begin{eqnarray*}
\rho_k(S) = \rho_k(S;X) = \frac{1}{n}\left|\sum_{i=1}^{n}c_{i1}\times c_{i2} \times ... \times c_{ik}\right|\,,
\end{eqnarray*}
with $c_{uv}$ being the $u$th element of $\mathbf{c}_v$. It is a measure of aliasing among the columns in $S$ and satisfies $0\leq \rho_k(S) \leq 1$. The quantity $\rho_k(S)$ can be thought of as the \textit{correlation} between two vectors formed as the products of exhaustive and mutually exclusive subsets of $S$. For example, $\rho_4(S)$ is the correlation between each one factor column in $S$ and the products of the remaining three columns, and also the correlation between each product of pairs of columns. 

We gather the aliasing indices into groups based on the number of columns, namely the {\it aliasing index groups} denoted as $\alpha_k$, such that for $k=1,2,\dots,m$,
\begin{eqnarray*}
\alpha_k = \left \{\rho_k(S):\rho_k(S)\neq 0  \right \},~k = 1,2,...,m\,.
\end{eqnarray*}
That is, $\alpha_k$ is a set that consists of nonzero aliasing indices obtained from the subsets that contain $k$ columns. If $\alpha_1$ is not the empty set, the design is not balanced, and if $\alpha_2$ is not empty, at least one pair of main effects is partially aliased (as must be the case in an SSD). If $\alpha_k$ is empty for $k>2$, no main effects are biased by $(k-1)$th-order interaction effects or, equivalently, no interaction involving $q<k$ factor columns is aliased with an interaction involving $k-q$ columns.

The summary of effect aliasing structure (SEAS) summarizes the information provided by the resolution, aberration and effect correlations but unlike a criterion such as $E(s^2)$-optimality, it does not try to  condense this information to a single number. There are three components in SEAS, abbreviated as MAP: (i) the \textbf{m}aximum dependency aliasing pattern; (ii) the \textbf{a}verage square aliasing pattern; and (iii) the \textbf{p}airwise dependency ratio. In this paper we consider only balanced SSDs, but these definitions readily extend to non-balanced designs. 

\begin{definition} 
\textbf{M-pattern}: For a design $\mathbf{X}$ with $n$ runs and $m$ factors, we define the maximum dependency aliasing pattern (M-pattern) as
\begin{eqnarray*}
\textsc{E}^{M}(\mathbf{X}) = \left ( e^M_1,e^M_2,...,e^M_m \right )\,,
\end{eqnarray*}
where
\begin{eqnarray*}
e^M_k = k + \frac{\max\alpha_k}{10}\,,\qquad k=1,\ldots,m\,,
\end{eqnarray*}
and $\max\alpha_k = \max_{\alpha_k} \rho_k(S)$ is the maximum aliasing index in the set $\alpha_k$. If $\alpha_k$ is the empty set, we define $\max\alpha_k=0$.
\end{definition}

Each entry in the M-pattern represents the largest (i.e. worst-case) correlation for each level of aliased effects. For example, $10(e^M_1-1)$ gives the maximum pairwise correlation between the intercept and the individual factor columns; $10(e^M_2-2)$ gives the maximum correlation between pairs of factor columns (corresponding to aliasing between main effects). Consider two SSDs, $\mathbf{X}_1$ and $\mathbf{X}_2$. If the first $k-1$ entries of their M-patterns are equal, and $\mathbf{X}_1$ has a lower $e^M_k$ than $\mathbf{X}_2$, then design $\mathbf{X}_1$ has a better M-pattern than $X_2$. It implies that $\mathbf{X}_1$ has lower worst-case partial aliasing than design $\mathbf{X}_2$, in that $\mathbf{X}_2$ has higher maximum correlation than $\mathbf{X}_1$ for between factor (main effect) columns and products of $k-1$ columns (corresponding to interactions involving $k-1$ factors) whilst the two designs have equivalent worst-case correlations for products of $k-1$ or fewer columns. If $\mathbf{X}_1$ has the best M-pattern among all SSDs of the same size, then $\mathbf{X}_1$ has the minimum M-pattern.

\begin{definition} 
\textbf{A-pattern}: For a design $X$ with $n$ runs and $m$ factors, we define the average square aliasing pattern (A-pattern) as
\begin{eqnarray*}
\textsc{E}^{A}(\mathbf{X}) = \left ( e^A_1,e^A_2,...,e^A_m \right )\,,
\end{eqnarray*}
where
\begin{eqnarray*}
e^A_k = k + \frac{\overline{\alpha^2_k}}{10}\,,\qquad k=1,2,\dots,m\,,
\end{eqnarray*}
and $\overline{\alpha^2_k} = \sum_{\alpha_k} \rho_k(S)^2/|\alpha_k|$ represents the average values of the squared aliasing indices in the set $\alpha_k$, with $|\alpha_k|$ is the cardinality of the set $\alpha_k$. If $\alpha_k$ is the empty set, we define $\overline{\alpha^2_k}=0$.
\end{definition}

Each entry in the A-pattern represents the average-squared correlation for each level of aliased effects. For example, $10(e^A_4-4)$ gives the average correlation between pairs of products of two columns, corresponding to aliasing between two-factor interactions or, equivalently, between main effects and three-factor interactions. The comparison for two SSDs with the A-pattern is the same as the case of the $M$-pattern. If a SSD has the best A-pattern among all SSDs of the same size, then it has the minimum A-pattern.

\begin{definition} 
\textbf{P-pattern}: For a design $X$ with $n$ runs and $m$ factors, we define the pairwise dependency ratio pattern (P-pattern) as
\begin{eqnarray*}
\textsc{E}^{P}(\mathbf{X}) = \left ( e^P_1,e^P_2,...,e^P_m \right )\,,
\end{eqnarray*}
where
\begin{eqnarray*}
e^P_k = k + \frac{|\alpha_k|/\binom{m}{k}}{10}\,,\qquad k=1,2,\dots,m\,.
\end{eqnarray*}
\end{definition}

Each entry in the P-pattern gives the percentage of nonzero aliasing indices over all possible $\binom{m}{k}$ combinations of factor columns. For example, $10(e^P_4-4)$ gives the proportion of pairs of two-factor interactions, or equivalently the proportion of pairs of main effects and three-factor interactions, which are partially aliased. For two SSDs $\mathbf{X}_1$ and $\mathbf{X}_2$, if the first $k-1$ entries of their P-patterns are equivalent and $\mathbf{X}_1$ has a lower $e^M_k$ than $\mathbf{X}_2$, then $\mathbf{X}_1$ has a better P-pattern than $\mathbf{X}_2$.

For each pattern, every entry consists of a positive digit and a decimal part. The positive digit refers to the value $k$. The decimal part, ranging from $0.000$ to $0.100$, represents the percentage from $0\%$ to $100\%$ of the pattern's focus (maximum, average or proportion). We include a division by 10 in the second term of each of $e^M_k$, $e^A_k$ and $e^P_k$ to avoid mixing  of the positive digit and the decimal part when the percentage reaches $100\%$.

\section{Theoretical relationships to traditional criteria}\label{sec:theory}
In this section, we outline the relationship between SEAS and existing quantities for assessing effect aliasing in non-regular fractional factorial designs (FFDs) and SSDs, namely the generalized resolution (GR), generalized wordlength pattern (GWLP) and $E(s^2)$-citerion. We demonstrate how each of these are actually summaries of the more detailed SEAS patterns.

Following \citet{PX2009}, we define the GR of a FFD $\mathbf{X}$ with $n$ factors and $m$ runs as
\begin{equation}\label{eq:gr}
GR(\mathbf{X}) = r+1-\max_{|S|=r}\rho_r(S)\,,
\end{equation}
where $r$ is the smallest integer such that $\max_{|S|=r} \rho_r(S)>0$. A direct relationship between the GR and the SEAS M-pattern is stated in the following theorem.

\begin{thm} \label{thm:SEAStoGR}
For a design $\mathbf{X}$ with $n$ factors and $m$ runs, $GR(\mathbf{X}) = r+1-10(e^M_r-r)$, where $r=\min\{k: e^M_k-k \neq 0\}$ for $k=1,\dots,m$.
\end{thm}
\begin{proof}
From the definition of the aliasing index set we have $\max_{|S|=r}\rho_r(S) = \max_r \alpha_r$, and from the definition of the M-pattern,  $e^M_r = r + \frac{\max\alpha_r}{10}$. Substituting these quantities into~\eqref{eq:gr} completes the proof.
\end{proof}

Theorem~\ref{thm:SEAStoGR} shows that the GR is just a particular one-number summary of the M-Pattern of the SEAS. For a non-supersaturated FFD, the GR indicates the minimum order of the interaction effects that bias estimators of the main effect. However, for a SSD, the integer part of the GR will always be two due to partial aliasing between main effects. However, the M-pattern provides detailed information on the worst-case aliasing between the main effects and all higher-order interactions. It can provide differentiation between two SSDs (or two FFDs) with the same GR. 

We define the GWLP of a FFD $\mathbf{X}$ with $n$ factors and $m$ runs as a vector with elements $A_k$, for $k=1,\dots,m$:
\begin{eqnarray*}
A_k(\mathbf{X}) = \sum_{|S|=k}(\rho_k(S))^2\,.
\end{eqnarray*}
Each $A_k$ is a generalization of the number of words of length $k$ in the defining relation of a regular FFD (\citealp{WH2009}, p.~217). The GWLP can be constructed from the SEAS A- and P-patterns via the following theorem.

\begin{thm} \label{thm:SEAStoGWLP}
For a design $\mathbf{X}$ with $n$ factors and $m$ run, $A_k(\mathbf{X})=100\binom{m}{k}(e^A_k-k)(e^P_k-k)$, where $A_k$ is the $k$th entry of the GWLP for $k=1,\dots,m$.
\end{thm}
\begin{proof}
From the definition of the aliasing index set we have $A_k=\sum_k \alpha_k^2$, which can be rewritten as $A_k=\bar{\alpha_k^2} \cdot |\alpha_k|$. From the definition of the P-pattern, $|\alpha_k|=10(e^P_k-k)\binom{m}{k}$, and from the definition of the A-pattern, $\bar{\alpha_k^2}=10(e^A_k-k)$. Thus, $A_k=10(e^A_k-k) \cdot 10(e^P_k-k) \cdot \binom{m}{k}=100\binom{m}{k}(e^A_k-k)(e^P_k-k)$, completing the proof.
\end{proof}

Theorem~\ref{thm:SEAStoGWLP} shows that each entry of the GWLP is proportional to the product of corresponding (shifted) entries of the SEAS A- and P-patterns. It implies that the number of words of length $k$ can be decomposed into information on the average effect aliasing and the pairwise dependency ratio. Two designs with the same GWLP may still have discrepancies in A-patterns and P-patterns, and these discrepancies may help in differentiating two designs for different purposes, as demonstrated in Section~\ref{sec:comp}.

Finally, we show the relationship between $E(s^2)$ and SEAS.

\begin{thm} \label{thm:SEAStoEs2}
For a design $\mathbf{X}$ with $n$ factors and $m$ runs, $E(s^2)=100n^2(e^A_2-2)(e^P_2-2)$.
\end{thm}
\begin{proof}
From the definition of $A_2$, we have $A_2=\sum_{|S=2|}\rho_2(S)^2=n^{-2}\sum_{i<j}s^2_{ij}$. Then the $E(s^2)$ value can be rewritten as $E(s^2)=n^2 A_2 / \binom{m}{2}$. From Theorem~\ref{thm:SEAStoGWLP}, we have $A_2=100\binom{m}{2}(e^A_2-2)(e^P_2-2)$, completing the proof.
\end{proof}

Expressing the $E(s^2)$ value as a function of only $e^A_2$ and $e^P_2$ reinforces the fact that this criterion only considers relationships between pairs of main effects in the design, assuming interaction effects to be negligible. In contrast, SEAS provides a comprehensive description on the relationships between all main effects and other main effects, and between all interaction effects. Hence, SEAS can discriminate between sets of good, or optimal, $E(s^2)$ designs is terms of their interaction aliasing properties. 

\section{Comparison of SSDs using the SEAS}\label{sec:comp}

\begin{table}
\caption{\label{tab:summary}Summary of design performance for the five designs discussed in Section~\ref{sec:comp}.}
\begin{center}
\begin{tabular}{c@{\hskip 3ex}cc@{\hskip 3ex}ccc@{\hskip 3ex}c}
\hline
& $E(s^2)$ & GR & GWLP &  &  SEAS & \\
Design &  &  & $(A_2,A_3)$ & $(e^M_2, e^M_4)$ & $(e^A_2,e^A_3)$ & $e^P_3$ \\
\hline
$D_1$ & 7.921	& 2.29 & (10.224, 141.713) & (2.0714, 4.0714) & (2.0040, 3.0128) & 3.0621 \\
$D_2$ & 7.921	& 2.29 & (10.224, 140.730) & (2.0714, 4.1000) & (2.0040, 3.0129) & 3.0615 \\
$D_3$ & 7.921	& 2.57 & (10.224, 141.712) & (2.0429, 4.1000) & (2.0040, 3.0136) & 3.0589 \\
$D_{Lin}$ & 7.921 & 2.57 & (10.224, 141.473) & (2.0429, 4.1000) & (2.0040, 3.0131) & 3.0612 \\
$D_{SIB}$	 & 7.415 & 2.57	& (9.571, 142.854) & (2.0429, 4.1000) & (2.0038, 3.0132) & 3.0610 \\
\hline
\end{tabular}
\end{center}
\end{table}

This section applies the SEAS measures to assess and compare five different SSDs with 14 runs and 23 factors. These designs are listed in Appendix~\ref{app:comp} and labelled $D_1$, $D_2$, $D_3$, $D_{Lin}$ and $D_{SIB}$. The first three designs are included for demonstration purposes. Design $D_{Lin}$ is from \citet{L1993}, noting that the design in that paper for 24 factors has two identical factor columns. Design $D_{SIB}$ is $E(s^2)$-optimal, found via the swarm intelligence-based algorithm of \citet{PCWW2016}. The properties of each design, $E(s^2)$, GR, GWLP and the SEAS, are given in Appendix~\ref{app:comp}, with Table~\ref{tab:summary} summarizing the quantities that will be discussed in this section.

\textbf{Comparison 1: GR and M-patterns.} Designs $D_1$ and $D_2$ share the same $GR=2.29$, implying that the worst-case aliasing among all main effect pairs is the same in both designs. The SEAS M-pattern can be used to distinguish these two designs, with the patterns first difference occurring for $e^M_4$, with $e^M_4=4.0714$ for $D_1$ and $e^M_4=4.1000$ for $D_2$. This element of the M-pattern measures aliasing between combinations of four factors, that is between a main effect and a three-factor interaction or between pairs of two-factor interactions. Design $D_1$ has no complete aliasing between these effects, whereas for design $D_2$ there is at least one pair of two-factor interactions that are fully aliased (or equivalently at least one main effect fully aliased with one three-factor interaction). Notice that although designs $D_3$, $D_{Lin}$ and $D_{SIB}$ all have lower values of $e^M_2$ than $D_1$, and therefore lower worst-case aliasing between pairs of main effects, these three designs also have $e^M_4=4.1000$ and hence have complete aliasing between at least one main effect and one three-factor interaction. 

\textbf{Comparison 2: GWLP, A-patterns and P-patterns.} Design $D_2$ is preferred to design $D_3$ under the criterion of GWLP, as $A_1$ and $A_2$ are both equal for the two designs, and $A_3$ is lower for $D_2$ than $D_3$. The A- and P-patterns can be used to further explore the differences between these designs. Under the A-pattern, $D_2$ remains a better choice ($e^A_3=3.0129$) than $D_3$ ($e^A_3=3.0136$). However, under the p-pattern, this preference is reversed, with $D_3$ ($e^P_3=3.0589$) slightly outperforming $D_2$ ($e^P_3=3.0615$). Design $D_2$ has slightly lower average partial aliasing between main effects and two-factor interactions, but design $D_3$ has a slightly lower proportion of pairs of main effects and three-factor interactions partially aliased. Such tensions are not uncommon in highly fractionated nonregular designs, and the SEAS measures allow these differences to be identified and quantified.

\textbf{Comparison 3: An Evaluation of $D_{Lin}$.} \citet{L1993} constructed a 14-run design as a half-fraction of a 28-run Plackett-Burman design. This design, and also $D_3$, has a higher GR (2.57) than either of designs $D_1$ or $D_2$ (both 2.29). For the GWLP, $D_{Lin}$ has a lower $A_3$ (141.473) than either $D_1$ (141.713) or $D_3$ (141.712) but higher than $D_2$, which has the lowest $A_3$ amongst these four designs (140.730). Design $D_{Lin}$ would appear preferable to either $D_1$ or $D_2$ for the estimation of main effects in the absence of any two-factor, or higher-order, interactions.

We can use the SEAS measures to compare these designs in more detail, and evaluate the performance of $D_{Lin}$ when some interactions are present. Firstly, the higher GR of $D_{Lin}$ leads to this design having a better M-pattern than either $D_1$ or $D_2$, with $e_2^M=2.0429$ against 2.0714. However, it should again be noted that $D_{Lin}$ has  $e_4^M=4.1000$ meaning that, as with $D_2$ and $D_3$, there full aliasing between at least one main effect and three-factor interaction. Designs $D_{Lin}$ and $D_3$ are essentially inseparable under the M-pattern, with the first difference occurring for $e^M_{21}$ (21.0571 for $D_{Lin}$ and 21.0857 for $D_3$), corresponding to differences in maximum partial aliasing between main effects and very high-order interactions.

Secondly, the four designs have the same values for $e^A_2=2.0040$ and $e^P_2=2.1000$, which is reflected in the equality of their $E(s^2)$ values. However, there are some interesting differences in the next values for both the A-and P-patterns. Although $D_{Lin}$ has a better $A_3$ entry of the GWLP than either $D_1$ or $D_3$,  these two designs are better choices in terms of the A-pattern and P-pattern, respectively. In contrast, $D_2$ has a better $A_3$ entry of the GWLP than $D_{Lin}$, but $D_{Lin}$ is a better choice under the P-pattern if a design is required with a minimum number of columns with non-zero partial aliasing. This illustrates the additional information obtained through the refinement of the GWLP into the A- and P-Patterns when choosing a good design for a specific purpose.

\textbf{Comparison 4: An Evaluation of $D_{SIB}$.} Design $D_{SIB}$ is $E(s^2)$-optimal and obtained by \citet{PCWW2016} using a swarm intelligence algorithm. It has a lower $E(s^2)$ value than the other four designs, reflected in its slightly smaller value of $e_2^A$ (2.0038 vs 2.0040 for the other designs). It also has a better GWLP, with a lower value of $A_2$ (9.571 vs 10.224). $D_{SIB}$ is a better design than the others for estimation of main effects in the absence of any two-factor, or higher-order, interactions. 

However, if interactions are considered, we again see some contradictions. $D_{SIB}$ also has $e_4^M=4.1000$, meaning there is full aliasing between some pairs of main effects and three-factor interactions (or between some pairs of two-factor interactions). The average aliasing between main effects and two-factor interactions is also higher for this design than for $D_1$, $D_2$ or $D_{Lin}$ ($e_3^A=3.0132$ vs 3.0128, 3.0129, 3.0131, respectively). It also has a higher portion of pairs of main effects and two-factor interactions partially aliased than design $D_3$ ($e_3^P=3.0610$ vs 3.0589), and the highest value of $A_3$ in the GWLP for all the designs. This illustration again shows that a design with the best GWLP cannot guarantee to be the best choice in either the average-squared correlation for each level of aliased effects (A-pattern) or the percentage of nonzero aliasing indices (P-pattern).

\section{Effect-SEAS for assessing the aliasing between individual factorial effects in a design}\label{sec:col}

Several authors have demonstrated how the aliasing properties of individual columns of a design can influence the effectiveness of, for example, variable selection through the assignment of factors to columns; see, for example, \citet{MW2010}. SEAS can also be used to describe the aliasing implied by a given design between an individual main effect and other factorial effects, and hence can inform the assignment of those factors thought more likely to be important to the columns with better properties. We call this application of SEAS \textit{Effect-SEAS}. In this section we only consider main effects but the definitions and applications can easily be extended to interaction effects. We start by defining individual alias index sets for each main effect.

Assume a design $\mathbf{X}$ with $n$ runs and $m$ factors. Let $\mathbf{x}_l$ be the column of $\mathbf{X}$ corresponding to the main effect of interest, for $1 \leq l \leq m$, and denote as $\mathbf{X}_{-l}$ the reduced design with $n$ runs and $m-1$ columns, where $\mathbf{x}_l$ is removed from $\mathbf{X}$. Let $S=\{\mathbf{c}_1,\dots,\mathbf{c}_{k-1}\}$ be a subset of $k-1$ columns of $\mathbf{X}_{-I}$, with $k\leq m-1$ and $\mathbf{c}_i=\mathbf{x}_j$ for some $j\in\{1,\ldots,m \}\setminus\{l\}$. Then the effect aliasing index of $\mathbf{x}_l$ with the $k-1$ effects in $S$, denoted as $\rho_k(\mathbf{x}_l | S)$, is defined as
\begin{eqnarray*}
\rho_k(\mathbf{x}_l | S) = \frac{1}{n}\left|\sum^n_{i=1} x_{il} \times c_{i1} \times \dots \times c_{i{k-1}} \right|\,,
\end{eqnarray*}
with $x_{uv}, c_{uv}$ being the $u$th element of $\mathbf{x}_v,\mathbf{c}_v$, respectively.
Next, we define the effect aliasing index group for $\mathbf{x}_l$ as
\begin{eqnarray*}
\alpha_k(\mathbf{x}_l)=\{\rho_k(\mathbf{x}_l | S):\rho_k(\mathbf{x}_l | S) \neq 0\}\,,\, k=2,\dots,m;\,k \neq l\,.
\end{eqnarray*}

We can now define the M-, A- and P-patterns for the main effect corresponding to $\mathbf{x}_l$ as follows.

\begin{definition}
For a design $\mathbf{X}$ with $n$ runs and $m$ factors, let $\mathbf{x}_l$ be column of $\mathbf{X}$ corresponding to the main effect of interest and $\mathbf{X}_{-l}$ be the reduced design with $\mathbf{x}_l$ removed. We define the effect specific maximum dependency aliasing pattern (M-pattern), average square aliasing pattern (A-pattern) and pairwise dependency ratio pattern (P-pattern) of $\mathbf{x}_l$ as
\begin{equation*}
E^M(\mathbf{x}_l |\mathbf{X}_{-l})=[e^M_2(\mathbf{x}_l),\dots, e^M_m(\mathbf{x}_l)]\,,
\end{equation*}
\begin{equation*}
E^A(\mathbf{x}_l | \mathbf{X}_{-l})=[e^A_2(\mathbf{x}_l), \dots, e^A_m(\mathbf{x}_l)]\,,
\end{equation*}
\begin{equation*}
E^P(\mathbf{x}_l | \mathbf{X}_{-l})=[e^P_2(\mathbf{x}_l), \dots, e^P_m(\mathbf{x}_l)]\,,
\end{equation*}
where $e^M_k(\mathbf{x}_l)=k+\frac{\max \alpha_k(\mathbf{x}_l)}{10}$, $e^A_k(\mathbf{x}_l)=k+\frac{\bar{\alpha_k}(\mathbf{x}_l)}{10}$ and $e^P_k(\mathbf{x}_l)=k+\frac{|\alpha_k(\mathbf{x}_l)|/\binom{m}{k}}{10}$, respectively, and $\max \alpha_k(\mathbf{x}_l) = \max_{\alpha_k(\mathbf{x}_l)}\rho_k(\mathbf{x}_l | S)$ and $\bar{\alpha_k}(\mathbf{x}_l) = \sum_{\alpha_k(\mathbf{x})_l}\rho_k(\mathbf{x}_l | S)/|\alpha_k(\mathbf{x})_l|$. If $\alpha_k(\mathbf{x}_l)=\emptyset$, we define $\max \alpha_k(\mathbf{x}_l) =0$.
\end{definition}

We demonstrate the use of the Effect-SEAS for individual main effects using design $D_{SIB}$ with $n=14$ runs and $m=23$ factors. The Effect-SEAS patterns for this design are given in Appendix~\ref{app:eseas}.

\textbf{Effect-SEAS 1: M-pattern analysis} (Table~\ref{tab:eseasM}). All columns have $e^M_2(\mathbf{x}_l)=2.0429$, and hence we first compare columns using $e^M_3(\mathbf{x}_l)$. Columns 1, 6, 9, 16, 18 and 23 have $e^M_3(\mathbf{x}_l)=3.0571$; the other columns have $e^M_3(\mathbf{x}_l)=3.0857$. Hence, these six columns have maximum correlation with products of any two columns which is 33\% smaller than the other columns. \citet{MW2010} showed that similar differences to this in column correlations can result in substantial differences in the power to detect active effects. If more than six active factors are anticipated, columns 4, 5, 12, and 17 should be avoided, as factors assigned to these columns would have main effects completely aliased with at least one three-factor interaction. 

\textbf{Effect-SEAS 2: A-pattern analysis} (Table~\ref{tab:eseasA}) Columns 8, 12 and 23 have the lowest $e^A_2(\mathbf{x}_l)=2.0028$, implying that the corresponding three main effects will have the least partial aliasing on average with other main effects. These columns can further be differentiated using $e^A_3(\mathbf{x}_l)$: column 12 is the best ($e^A_3(\mathbf{x}_l)=3.0135)$, followed by column 23 ($e^A_3(\mathbf{x}_l)=3.0136$) and column 8 ($e^A_3(\mathbf{x}_l)=3.0139$). A complete ranking under the A-pattern of all 23 columns in ascending order is: 12, 23, 8, 18, 21, 22, 1, 7, 11, 6, 3, 5, 13, 15, 16, 10, 17, 9, 14, 2, 20, 19, 4.

\textbf{Effect-SEAS 3: P-pattern analysis} (Table~\ref{tab:eseasP}) All columns have $e^P_2(\mathbf{x}_l)=2.1$ and hence again we compare using $e^M_3(\mathbf{x}_l)$. Column 19 has the smallest $e^P_3(\mathbf{x}_l)=3.0567$, implying that the corresponding main effect would be aliased with the smallest proportion of two-factor interactions. A complete ranking under the P-pattern of all 23 columns in ascending order is: 19, 20, 2, 14, 13, 5, 8, 4, 23, 3, 11, 12, 6, 1, 7, 9, 17, 22, 16, 10, 21, 15, 18.

Whilst the rankings under the patterns do not agree, it is clear that some columns rank highly under all three (e.g. columns 8 and 23) and should be consider first when assigning to columns those factors thought mostly likely to be active. Other columns rank poorly under all three patterns (e.g. columns 4 and 17) and should be reserved for those factors thought least likely to be important.

\section{Concluding remarks}\label{sec:disc}

Commonly applied criteria for the selection of SSDs generally reduce design performance down to a single number, typically summarising the aliasing structure between main effects. If interactions effects may be present, they can bias main effect estimators and hence lead to incorrect conclusions being drawn about which factors are important. Similarly, a study of the aliasing properties of an SSD at a finer graduation can allow discrimination between designs that have identical performance under criteria such as $E(s^2)$-optimality. It also allows assessment of the impact of the choice of assignment of factors to columns. 

In this paper we have proposed and demonstrated new descriptive statistics, SEAS, for the characterization of the aliasing structure of an SSD, or non-regular fractional factorial designs more generally. We have shown how the three MAP components of SEAS, (1) Maximum dependency aliasing pattern; (2) Average square aliasing pattern; and (3) Pairwise dependency ratio, generalize the summarizes provide by $E(s^2)$-optimality, generalized resolution and generalized wordlength patterns. The application of SEAS has been demonstrated on a set of exemplar designs with 23 factors and 14 runs, including two optimal or efficient designs suggested from the literature. There were differences between the ranking of designs under the components of SEAS compared to the traditional criteria, mainly due to these traditional criteria being either a special case of SEAS or aggregating information from the SEAS components.    

We also proposed Effect-SEAS for describing the correlation of an individual factorial effect with other effects in an experiment. This criterion allows the ranking of columns for the assignment of factors, using the extent of their correlation with other columns and products of columns (corresponding to partial aliasing with main effects and interactions). When demonstrated on an $E(s^2)$-optimal design, Effect-SEAS quantified the differences between columns in such a way that informed decisions on factor to column assignments could be made.

A main limitation to SEAS is the rapidly increasing length of the MAP patterns as the number of factors in the design increases, and the computational cost of obtaining them. A simple solution is to appeal to effect sparsity and hierarchy (\citealp{WH2009}, pp.~172--173) and only calculate patterns up to the fourth or fifth terms, corresponding to aliasing between four factors and five factors. 

 SEAS is a very general purpose tool,  and future work could apply the summaries to more general nonregular designs, and use Effect-SEAS to choose factor to column assignments to minimize the resulting aliasing of interaction effects. 

\section*{Acknowledgements}
Frederick Phoa was supported by a Career Development Award from Academia Sinica (Taiwan), grant number 103-CDA-M04, and the Ministry of Science and Technology (Taiwan), grant numbers 104-2118-M-001-016-MY2, 105-2118-M-001-007-MY2 and 105-2911-I-001-516-MY2. David Woods was supported by Fellowship EP/J018317/1 from the UK Engineering and Physical Sciences Research Council and by an award from the UK Royal Society International Exchanges scheme.

\appendix
\section{The SEAS of Five SSDs in Section~\ref{sec:comp}}\label{app:comp}

Each design is given in the form of a design vector. To construct the design matrix, each entry of the vector should be converted into its binary code to form a column of the matrix. For example, the first entry in $D_{SIB}$ is $1207$. For a design with $n=14$ runs, the binary code of $1207$ is $000010010110111$, so the first column of this design is $(-1,-1,-1,-1,+1,-1,-1,+1,-1,+1,+1,-1,+1,+1,+1)^{\T}$. 

\singlespacing

\begin{table}
\caption{\label{tab:d1_dseas}Summary of the aliasing structure for design $D_1$}
  \begin{center}
  \begin{tabular}{@{}ll@{}}
  \toprule[1pt]
 & $E(s^2)=7.921$; $GR=2.29$ \\
  \midrule[1pt]
  Design & $(127,~953,~1507,~1906,~2524,~2794,~3253,~5356,~7363,~7508,~8918,~9464,$\\
  vector & $~9614,~9937,~10053,~10598,~10721,~11291,~11430,~11880,~12405,~12619,$ \\
  &$~12722)$\\
  \midrule
  M-pattern & $(1.0000, 2.0714, 3.0857, 4.0714, 5.0857, 6.1000, 7.0857, 8.1000, 9.0857, 10.1000,$\\
  & $~11.0857, 12.1000, 13.0857, 14.1000, 15.0857, 16.1000, 17.0857, 18.1000, 19.0857,$\\ 
  & $~20.0714,21.0857, 22.1000, 23.0000)$\\
  \midrule
  A-pattern &$(1.0000, 2.0040, 3.0128, 4.0075, 5.0121, 6.0071, 7.0124, 8.0072, 9.0123, 10.0071,$ \\ 
  & $~11.0123, 12.0071, 13.0123, 14.0071, 15.0123, 16.0071, 17.0123, 18.0072, 19.0124,$ \\ 
  &$~20.0070, 21.0121, 22.0084, 23.0000)$\\
  \midrule
  P-pattern &$(1.0000, 2.1000, 3.0621, 4.1000, 5.0575, 6.1000, 7.0583, 8.1000, 9.0580, 10.1000,$\\ 
  & $~11.0581, 12.1000, 13.0580, 14.1000, 15.0582, 16.1000, 17.0580, 18.1000, 19.0579,$ \\
  & $~20.1000, 21.0597, 22.1000, 23.0000)$\\
  \midrule
  GWLP & $(0.000, 10.224, 141.713, 661.026, 2340.610, 7153.104, 17663.815, 35086.870,$ \\ 
  & $~58129.518, 81686.312, 96833.138, 96551.890, 81539.363, 58420.913, 35107.249,$ \\ 
  & $~17477.242, 7186.267, 2417.681, 635.413, 123.332, 18.286, 1.939, 0.000)$\\
  \bottomrule[1pt]
  \end{tabular}
  \end{center}
  \end{table}
  
  \begin{table}
\caption{\label{tab:d2_dseas}Summary of the aliasing structure for design $D_2$}
  \begin{center}
  \begin{tabular}{@{}ll@{}}
  \toprule[1pt]
 & $E(s^2)=7.921$; $GR=2.29$ \\
  \midrule[1pt]
  Design & $(926,~1769,~1877,~2414,~2545,~2663,~2771,~3644,~4002,~5005,~5242,~5461,$ \\
  vector & $~6580,~7112,~7828,~8918,~9017,~9712,~10460,~11405,~12100,~13028,~14897)$\\
  \midrule
  M-pattern & $(1.0000, 2.0714, 3.0857, 4.1000, 5.0857, 6.1000, 7.0857, 8.1000, 9.0857, 10.1000,$ \\ 
  & $~11.0857, 12.1000, 13.0857, 14.1000, 15.0857, 16.1000, 17.0857, 18.1000, 19.0857,$ \\ 
  & $~20.1000, 21.0857, 22.0429, 23.0571)$\\
  \midrule
  A-pattern & $(1.0000, 2.0040, 3.0129, 4.0075, 5.0122, 6.0071, 7.0124, 8.0072, 9.0123, 10.0071,$ \\ 
  & $~11.0123, 12.0072,13.0123, 14.0071, 15.0123, 16.0072, 17.0122, 18.0071, 19.0123,$ \\ 
  & $~20.0073, 21.0113, 22.0056, 23.0327)$\\
  \midrule
  P-pattern & $(1.0000, 2.1000, 3.0615, 4.1000, 5.0573, 6.1000, 7.0583, 8.1000, 9.0580, 10.1000,$ \\ 
  & $~11.0582, 12.1000,13.0580, 14.1000, 15.0582, 16.1000, 17.0580, 18.1000, 19.0592,$ \\ 
  & $~20.1000, 21.0549, 22.1000, 23.1000)$\\
  \midrule
  GWLP & $(0.000, 10.224, 140.730, 663.682, 2345.820, 7134.934, 17655.412, 35145.708,$ \\ 
  & $~58127.110, 81583.346, 96862.012, 96687.098, 81480.367, 58314.678, 35159.664,$\\ 
  & $~17531.178, 7157.369, 2396.818, 646.201,128.557, 15.674, 1.286, 0.327)$\\
  \bottomrule[1pt]
  \end{tabular}
  \end{center}
  \end{table}
  
  \begin{table}
\caption{\label{tab:d3_dseas}Summary of the aliasing structure for design $D_3$}
  \begin{center}
  \begin{tabular}{@{}ll@{}}
  \toprule[1pt]
 & $E(s^2)=7.921$; $GR=2.57$ \\
  \midrule[1pt]
  Design & $(1739,~1894,~1897,~2263,~3239,~3667,~5011,~5181,~7396,~7956,~9457,~9558,$\\ 
  vector & $~9650,~10117,~10553,~10574,~10791,~11928,~12916,~12997,~13327,~15681,$ \\
  & $~15906)$\\
  \midrule
  M-pattern & $(1.0000, 2.0429, 3.0857, 4.1000, 5.0857, 6.1000, 7.0857, 8.1000, 9.0857, 10.1000,$ \\ 
  & $~11.0857, 12.1000,13.0857, 14.1000, 15.0857, 16.1000, 17.0857, 18.1000, 19.0857,$ \\ 
  & $~20.0714, 21.0857, 22.0714, 23.0000)$\\
  \midrule
  A-pattern &$(1.0000, 2.0040, 3.0136, 4.0075, 5.0121, 6.0071, 7.0123, 8.0072, 9.0123, 10.0071,$ \\ 
  & $~11.0123, 12.0071,13.0123, 14.0071, 15.0123, 16.0071, 17.0123, 18.0072, 19.0124,$ \\ 
  & $~20.0068, 21.0117, 22.0098, 23.0000)$\\
  \midrule
  P-pattern &$(1.0000, 2.1000, 3.0589, 4.1000, 5.0574, 6.1000, 7.0584, 8.1000, 9.0580, 10.1000,$ \\ 
  & $~11.0581, 12.1000,13.0581, 14.1000, 15.0582, 16.1000, 17.0578, 18.1000, 19.0584,$ \\ 
  & $~20.1000, 21.0597, 22.1000, 23.0000)$\\
  \midrule
  GWLP & $(0.000, 10.224, 141.712, 661.380, 2339.972, 7151.085, 17669.563, 35091.773,$ \\ 
  & $~58112.617, 81686.312, 96865.651, 96551.890, 81492.724, 58420.913, 35146.262,$ \\ 
  & $~17467.437, 7167.927, 2424.410, 640.648,121.049, 17.632, 2.265, 0.000)$\\
  \bottomrule[1pt]
  \end{tabular}
  \end{center}
  \end{table}
  
  \begin{table}
\caption{\label{tab:dlin_dseas}Summary of the aliasing structure for design $D_{Lin}$}
  \begin{center}
  \begin{tabular}{@{}ll@{}}
  \toprule[1pt]
 & $E(s^2)=7.921$; $GR=2.57$ \\
  \midrule[1pt]
  Design & See Lin (1993) \\
  vector & \\
  \midrule
  M-Pattern & $(1.0000, 2.0429, 3.0857, 4.1000, 5.0857, 6.1000, 7.0857, 8.1000, 9.0857, 10.1000,$ \\ 
  & $~11.0857, 12.1000,13.0857, 14.1000, 15.0857, 16.1000, 17.0857, 18.1000, 19.0857,$ \\ 
  & $~20.0714, 21.0571, 22.0429, 23.0286)$\\
  \midrule
  A-Pattern &$(1.0000, 2.0040, 3.0131, 4.0075, 5.0123, 6.0071, 7.0122, 8.0072, 9.0124, 10.0071,$ \\ 
  & $~11.0122, 12.0071,13.0124, 14.0071, 15.0122, 16.0071, 17.0123, 18.0072, 19.0121,$ \\ 
  & $~20.0071, 21.0127, 22.0070, 23.0082)$\\
  \midrule
  P-Pattern &$(1.0000, 2.1000, 3.0612, 4.1000, 5.0566, 6.1000, 7.0589, 8.1000, 9.0575, 10.1000,$ \\ 
  & $~11.0587, 12.1000, 13.0574, 14.1000, 15.0587, 16.1000, 17.0577, 18.1000, 19.0592,$ \\ 
  & $~20.1000, 21.0553, 22.1000, 23.1000)$\\
  \midrule
  GWLP & $(0.000, 10.224, 141.473, 661.380, 2342.624, 7151.085, 17655.475, 35096.676,$ \\ 
  & $~58149.599, 81663.431, 96805.278, 96592.452, 81555.532, 58380.054, 35099.289,$ \\ 
  & $~17504.210, 7188.022, 2405.904, 636.180, 126.272, 17.795, 1.612, 0.0816)$\\
  \bottomrule[1pt]
  \end{tabular}
  \end{center}
  \end{table}
  
  \begin{table}
\caption{\label{tab:dsib_dseas}Summary of the aliasing structure for design $D_{SIB}$}
  \begin{center}
  \begin{tabular}{@{}ll@{}}
  \toprule[1pt]
& $E(s^2)=7.415$; $GR=2.57$ \\
  \midrule[1pt]
  Design & $(1207,~1479,~1964,~2426,~2774,~3181,~4726,~5041,~5275,~5368,~5678,~6439,$ \\
  vector & $~6556, 7876,~8682,~8847,~9588,~10428,~11825,~12381,~12517,~13590,~15522)$\\
  \midrule
  M-pattern & $(1.0000, 2.0429, 3.0857, 4.1000, 5.0857, 6.1000, 7.0857, 8.1000, 9.0857, 10.1000,$ \\
  & $~11.0857, 12.1000, 13.0857, 14.1000, 15.0857, 16.1000, 17.0857, 18.1000, 19.0857,$ \\ 
   & $~20.1000, 21.0857, 22.0714, 23.0000)$\\
  \midrule
  A-pattern &$(1.000, 2.0038, 3.0132, 4.0075, 5.0121, 6.0071, 7.0124, 8.0072, 9.0122, 10.0071,$ \\ 
  & $~11.0124, 12.0071,13.0122, 14.0071, 15.0124, 16.0071, 17.0122, 18.0072, 19.0125,$ \\ 
  & $~20.0068, 21.0125, 22.0106, 23.0000)$\\
  \midrule
  P-pattern &$(1.0000, 2.1000, 3.0610, 4.1000, 5.0574, 6.1000, 7.0582, 8.1000, 9.0582, 10.1000,$ \\ 
  & $~11.0580, 12.1000,13.0582, 14.1000, 15.0581, 16.1000, 17.0580, 18.1000, 19.0584,$ \\ 
  & $~20.1000, 21.0542, 22.1000, 23.0000)$\\
  \midrule
  GWLP & $(0.000, 9.571, 142.854, 666.427, 2330.354, 7133.924, 17705.046, 35121.192,$ \\ 
  & $~58036.921, 81651.990, 96973.148, 96578.932, 81392.731, 58420.913, 35205.141,$ \\ 
  & $~17457.630, 7145.229, 2429.458, 645.7176, 119.578, 17.143, 2.429, 0.000)$\\
  \bottomrule[1pt]
  \end{tabular}
  \end{center}
  \end{table}
\clearpage

\section{The Effect-SEAS of $D_{SIB}$ in Section~\ref{sec:col}}\label{app:eseas}

\begin{longtable}{c|l}
\caption{Effect-SEAS for $D_{SIB}$: M-pattern}\label{tab:eseasM} \\
\centering
Column & M-Pattern \\
\midrule[1pt]
1  & 2.0429, 3.0571, 4.0714, 5.0857, 6.1, 7.0857, 8.1, 9.0857, 10.1, 11.0857, 12.1, \\*
    & 13.0857, 14.1, 15.0857, 16.1, 17.0857, 18.1, 19.0857, 20.1, 21.0857, 22.0714 \\
\midrule[1pt]
2  & 2.0429, 3.0857, 4.0714, 5.0857, 6.1, 7.0857, 8.1, 9.0857, 10.1, 11.0857, 12.1, \\*
    & 13.0857, 14.1, 15.0857, 16.1, 17.0857, 18.1, 19.0857, 20.1, 21.0857, 22.0714 \\
\midrule[1pt]
3  & 2.0429, 3.0857, 4.0714, 5.0857, 6.1, 7.0857, 8.1, 9.0857, 10.1, 11.0857, 12.1, \\*
    & 13.0857, 14.1, 15.0857, 16.1, 17.0857, 18.1, 19.0857, 20.1, 21.0857, 22.0714 \\
\midrule[1pt]
4  & 2.0429, 3.0857, 4.1000, 5.0857, 6.1, 7.0857, 8.1, 9.0857, 10.1, 11.0857, 12.1, \\*
    & 13.0857, 14.1, 15.0857, 16.1, 17.0857, 18.1, 19.0857, 20.1, 21.0857, 22.0429 \\
\midrule[1pt]
5  & 2.0429, 3.0857, 4.1000, 5.0857, 6.1, 7.0857, 8.1, 9.0857, 10.1, 11.0857, 12.1, \\*
    & 13.0857, 14.1, 15.0857, 16.1, 17.0857, 18.1, 19.0857, 20.1, 21.0857, 22.0714 \\
\midrule[1pt]
6  & 2.0429, 3.0571, 4.0714, 5.0857, 6.1, 7.0857, 8.1, 9.0857, 10.1, 11.0857, 12.1, \\*
    & 13.0857, 14.1, 15.0857, 16.1, 17.0857, 18.1, 19.0857, 20.1, 21.0857, 22.0714 \\
\midrule[1pt]
7  & 2.0429, 3.0857, 4.0714, 5.0857, 6.1, 7.0857, 8.1, 9.0857, 10.1, 11.0857, 12.1, \\*
    & 13.0857, 14.1, 15.0857, 16.1, 17.0857, 18.1, 19.0857, 20.1, 21.0857, 22.0714 \\
\midrule[1pt]
8  & 2.0429, 3.0857, 4.0714, 5.0857, 6.1, 7.0857, 8.1, 9.0857, 10.1, 11.0857, 12.1, \\*
    & 13.0857, 14.1, 15.0857, 16.1, 17.0857, 18.1, 19.0857, 20.1, 21.0857, 22.0714 \\
\midrule[1pt]
9  & 2.0429, 3.0571, 4.0714, 5.0857, 6.1, 7.0857, 8.1, 9.0857, 10.1, 11.0857, 12.1, \\*
    & 13.0857, 14.1, 15.0857, 16.1, 17.0857, 18.1, 19.0857, 20.1, 21.0857, 22.0714 \\
\midrule[1pt]
10 & 2.0429, 3.0857, 4.0714, 5.0857, 6.1, 7.0857, 8.1, 9.0857, 10.1, 11.0857, 12.1, \\*
     & 13.0857, 14.1, 15.0857, 16.1, 17.0857, 18.1, 19.0857, 20.1, 21.0857, 22.0714 \\
\midrule[1pt]
11 & 2.0429, 3.0857, 4.0714, 5.0857, 6.1, 7.0857, 8.1, 9.0857, 10.1, 11.0857, 12.1, \\*
     & 13.0857, 14.1, 15.0857, 16.1, 17.0857, 18.1, 19.0857, 20.1, 21.0857, 22.0714 \\
\midrule[1pt]
12 & 2.0429, 3.0857, 4.1000, 5.0857, 6.1, 7.0857, 8.1, 9.0857, 10.1, 11.0857, 12.1, \\*
    & 13.0857, 14.1, 15.0857, 16.1, 17.0857, 18.1, 19.0857, 20.1, 21.0857, 22.0714 \\
\midrule[1pt]
13 & 2.0429, 3.0857, 4.0714, 5.0857, 6.1, 7.0857, 8.1, 9.0857, 10.1, 11.0857, 12.1, \\*
    & 13.0857, 14.1, 15.0857, 16.1, 17.0857, 18.1, 19.0857, 20.1, 21.0857, 22.0714 \\
\midrule[1pt]
14 & 2.0429, 3.0857, 4.0714, 5.0857, 6.1, 7.0857, 8.1, 9.0857, 10.1, 11.0857, 12.1, \\*
     & 13.0857, 14.1, 15.0857, 16.1, 17.0857, 18.1, 19.0857, 20.1, 21.0857, 22.0714 \\
\midrule[1pt]
15 & 2.0429, 3.0857, 4.0714, 5.0857, 6.1, 7.0857, 8.1, 9.0857, 10.1, 11.0857, 12.1, \\*
     & 13.0857, 14.1, 15.0857, 16.1, 17.0857, 18.1, 19.0857, 20.1, 21.0857, 22.0714 \\
\midrule[1pt]
16 & 2.0429, 3.0571, 4.0714, 5.0857, 6.1, 7.0857, 8.1, 9.0857, 10.1, 11.0857, 12.1,\\*
     & 13.0857, 14.1, 15.0857, 16.1, 17.0857, 18.1, 19.0857, 20.1, 21.0857, 22.0714 \\
\midrule[1pt]
17 & 2.0429, 3.0857, 4.1000, 5.0857, 6.1, 7.0857, 8.1, 9.0857, 10.1, 11.0857, 12.1, \\*
    & 13.0857, 14.1, 15.0857, 16.1, 17.0857, 18.1, 19.0857, 20.1, 21.0857, 22.0714 \\
\midrule[1pt]
18 & 2.0429, 3.0571, 4.0714, 5.0857, 6.1, 7.0857, 8.1, 9.0857, 10.1, 11.0857, 12.1, \\*
    & 13.0857, 14.1, 15.0857, 16.1, 17.0857, 18.1, 19.0857, 20.1, 21.0857, 22.0714 \\
\midrule[1pt]
19 & 2.0429, 3.0857, 4.0714, 5.0857, 6.1, 7.0857, 8.1, 9.0857, 10.1, 11.0857, 12.1, \\*
    & 13.0857, 14.1, 15.0857, 16.1, 17.0857, 18.1, 19.0857, 20.1, 21.0857, 22.0714 \\
\midrule[1pt]
20 & 2.0429, 3.0857, 4.0714, 5.0857, 6.1, 7.0857, 8.1, 9.0857, 10.1, 11.0857, 12.1, \\*
    & 13.0857, 14.1, 15.0857, 16.1, 17.0857, 18.1, 19.0857, 20.1, 21.0857, 22.0714 \\
\midrule[1pt]
21 & 2.0429, 3.0857, 4.0714, 5.0857, 6.1, 7.0857, 8.1, 9.0857, 10.1, 11.0857, 12.1, \\*
    & 13.0857, 14.1, 15.0857, 16.1, 17.0857, 18.1, 19.0857, 20.1, 21.0857, 22.0714 \\
\midrule[1pt]
22 & 2.0429, 3.0857, 4.0714, 5.0857, 6.1, 7.0857, 8.1, 9.0857, 10.1, 11.0857, 12.1, \\*
    & 13.0857, 14.1, 15.0857, 16.1, 17.0857, 18.1, 19.0857, 20.1, 21.0857, 22.0714 \\
\midrule[1pt]
23 & 2.0429, 3.0571, 4.0714, 5.0857, 6.1, 7.0857, 8.1, 9.0857, 10.1, 11.0857, 12.1, \\*
    & 13.0857, 14.1, 15.0857, 16.1, 17.0857, 18.1, 19.0857, 20.1, 21.0857, 22.0714 \\
\bottomrule[1pt]
\end{longtable}

\begin{longtable}{c|l}
\caption{Effect-SEAS for $D_{SIB}$: A-pattern}\label{tab:eseasA} \\
\centering
Column & A-Pattern \\
\midrule[1pt]
1  &2.0035,  3.0132,  4.0076,  5.0121,  6.0071,  7.0124,  8.0072, 9.0122, 10.0071, \\*
& 11.0124, 12.0071, 13.0122, 14.0071, 15.0124, 16.0071, 17.0122, 18.0072, \\* 
& 19.0124, 20.0068, 21.0122, 22.0102 \\
\midrule[1pt]
2  &2.0043,  3.0138,  4.0075,  5.0122,  6.0071,  7.0124,  8.0072,  9.0122, 10.0071, \\* 
& 11.0124, 12.0071, 13.0122, 14.0071, 15.0124, 16.0071, 17.0122, 18.0072, \\*
& 19.0124, 20.0068, 21.0125, 22.0102 \\
\midrule[1pt]
3  & 2.0035,  3.0135,  4.0076,  5.0110,  6.0071,  7.0124,  8.0072,  9.0122, 10.0071, \\* 
& 11.0124, 12.0071, 13.0122, 14.0072, 15.0124, 16.0071, 17.0122, 18.0072, \\* 
& 19.0125, 20.0067, 21.0123, 22.0109 \\
\midrule[1pt]
4  & 2.0050,  3.0132,  4.0074,  5.0123,  6.0071,  7.0123,  8.0072,  9.0122, 10.0071, \\*
& 11.0124, 12.0071, 13.0122, 14.0071, 15.0123  16.0071, 17.0122, 18.0072, \\* 
& 19.0124, 20.0070, 21.0129, 22.0087 \\
\midrule[1pt]
5  & 2.0035,  3.0139,  4.0076,  5.0110,  6.0071,  7.0124,  8.0072,  9.0122, 10.0071, \\* 
& 11.0124, 12.0071, 13.0122, 14.0072, 15.0124, 16.0071, 17.0121, 18.0072, \\* 
& 19.0126, 20.0067, 21.0123, 22.0109 \\
\midrule[1pt]
6  & 2.0035,  3.0134,  4.0076,  5.0120,  6.0071,  7.0124,  8.0072,  9.0122, 10.0071, \\* 
& 11.0124, 12.0071, 13.0122, 14.0072, 15.0124  16.0071, 17.0122, 18.0072, \\* 
& 19.0125, 20.0067, 21.0126, 22.0109 \\
\midrule[1pt]
7  & 2.0035,  3.0133,  4.0076,  5.0120,  6.0071,  7.0124,  8.0072,  9.0122, 10.0071, \\* 
& 11.0124, 12.0071, 13.0122, 14.0072, 15.0124, 16.0071, 17.0122, 18.0072, \\* 
& 19.0125, 20.0067, 21.0126, 22.0109 \\
\midrule[1pt]
8  & 2.0028,  3.0139,  4.0077,  5.0121,  6.0070,  7.0124,  8.0072,  9.0122, 10.0071, \\* 
& 11.0124, 12.0071, 13.0122, 14.0071, 15.0124, 16.0071, 17.0122, 18.0072, \\* 
& 19.0125, 20.0067, 21.0124, 22.0109 \\
\midrule[1pt]
9  & 2.0043,  3.0130,  4.0075,  5.0120,  6.0071,  7.0124,  8.0072,  9.0122, 10.0071, \\* 
& 11.0124, 12.0071, 13.0122, 14.0071, 15.0124, 16.0071, 17.0122, 18.0072, \\* 
& 19.0125, 20.0068, 21.0128, 22.0102 \\
\midrule[1pt]
10 & 2.0043,  3.0125,  4.0075,  5.0121,  6.0071,  7.0124,  8.0072,  9.0122, 10.0071, \\* 
& 11.0124, 12.0071, 13.0122, 14.0072, 15.0124, 16.0071, 17.0122, 18.0072, \\* 
& 19.0124, 20.0067, 21.0125, 22.0109 \\
\midrule[1pt]
11 & 2.0035,  3.0133,  4.0076,  5.0121,  6.0071,  7.0124,  8.0072,  9.0122, 10.0071, \\* 
& 11.0124, 12.0071, 13.0122, 14.0071, 15.0127, 16.0071, 17.0122, 18.0072, \\* 
& 19.0125, 20.0068, 21.0128, 22.0102 \\
\midrule[1pt]
12 & 2.0028,  3.0135,  4.0076,  5.0120,  6.0070,  7.0124,  8.0072,  9.0122, 10.0071, \\*
& 11.0124, 12.0071, 13.0122, 14.0071, 15.0124, 16.0071, 17.0122, 18.0072, \\* 
& 19.0126, 20.0067, 21.0122, 22.0109 \\
\midrule[1pt]
13 & 2.0035,  3.0139,  4.0076,  5.0122,  6.0071,  7.0125,  8.0072,  9.0122, 10.0071, \\* 
& 11.0124, 12.0071, 13.0122, 14.0072, 15.0124, 16.0071, 17.0122, 18.0072, \\* 
& 19.0126, 20.0067, 21.0124, 22.0109 \\
\midrule[1pt]
14 & 2.0043,  3.0138,  4.0075,  5.0120,  6.0071,  7.0124,  8.0072,  9.0122, 10.0071, \\* 
& 11.0124, 12.0071, 13.0122, 14.0071, 15.0123, 16.0071, 17.0122, 18.0072, \\* 
& 19.0125, 20.0068, 21.0128, 22.0102 \\
\midrule[1pt]
15 & 2.0043,  3.0121,  4.0075,  5.0120,  6.0071,  7.0124,  8.0072,  9.0122, 10.0071, \\* 
& 11.0124, 12.0071, 13.0122, 14.0071, 15.0123, 16.0071, 17.0122, 18.0072, \\*  
& 19.0125, 20.0068, 21.0128, 22.0102 \\
\midrule[1pt]
16 & 2.0043,  3.0125,  4.0075,  5.0120,  6.0071,  7.0124,  8.0072,  9.0122, 10.0071, \\* 
& 11.0124, 12.0071, 13.0122, 14.0071, 15.0123, 16.0071, 17.0122, 18.0072, \\* 
& 19.0124, 20.0068, 21.0128, 22.0102 \\
\midrule[1pt]
17 & 2.0043,  3.0127,  4.0075,  5.0121,  6.0071,  7.0124,  8.0072,  9.0122, 10.0071, \\* 
& 11.0124, 12.0071, 13.0122, 14.0071, 15.0124, 16.0071, 17.0122, 18.0072, \\* 
& 19.0125, 20.0068, 21.0124, 22.0102 \\
\midrule[1pt]
18 & 2.0035,  3.0117,  4.0076,  5.0120,  6.0071,  7.0124,  8.0072,  9.0122, 10.0071, \\* 
& 11.0124, 12.0071, 13.0122, 14.0071, 15.0124, 16.0071, 17.0122, 18.0072, \\* 
& 19.0125, 20.0068, 21.0125, 22.0102 \\
\midrule[1pt]
19 & 2.0043,  3.0142,  4.0075,  5.0121,  6.0071,  7.0124,  8.0072,  9.0122, 10.0071, \\* 
& 11.0124, 12.0071, 13.0122, 14.0072, 15.0124, 16.0071, 17.0122, 18.0072, \\* 
& 19.0125, 20.0067, 21.0125, 22.0109 \\
\midrule[1pt]
20 & 2.0043,  3.0141,  4.0075,  5.0121,  6.0071,  7.0124,  8.0072,  9.0122, 10.0071, \\* 
& 11.0124, 12.0071, 13.0122, 14.0072, 15.0124, 16.0071, 17.0122, 18.0072, \\* 
& 19.0125, 20.0067, 21.0123, 22.0109 \\
\midrule[1pt]
21 & 2.0035,  3.0125,  4.0076,  5.0120,  6.0071,  7.0124,  8.0072,  9.0122, 10.0071, \\* 
& 11.0124, 12.0071, 13.0122, 14.0072, 15.0124, 16.0071, 17.0122, 18.0072, \\* 
& 19.0125, 20.0067, 21.0124, 22.0109 \\
\midrule[1pt]
22 & 2.0035,  3.0129,  4.0076,  5.0120,  6.0071,  7.0124,  8.0072,  9.0122, 10.0071, \\* 
& 11.0124, 12.0071, 13.0122, 14.0072, 15.0124, 16.0071, 17.0122, 18.0072, \\* 
& 19.0126, 20.0067, 21.0123, 22.0109 \\
\midrule[1pt]
23 & 2.0028,  3.0136,  4.0076,  5.0120,  6.0070,  7.0124,  8.0072,  9.0122, 10.0071, \\* 
& 11.0124, 12.0071, 13.0122, 14.0071, 15.0124, 16.0071, 17.0122, 18.0072, \\* 
& 19.0125, 20.0067, 21.0124, 22.0109 \\
\bottomrule[1pt]
\end{longtable}

\begin{longtable}{c|l}
\caption{Effect-SEAS for $D_{SIB}$: P-pattern}\label{tab:eseasP} \\
\centering
Column & P-Pattern \\
1  & 2.1, 3.0610, 4.1, 5.0572, 6.1, 7.0583, 8.1,   9.0582, 10.1, 11.0580, 12.1,    \\*
   & 13.0582, 14.1, 15.0581, 16.1, 17.0580, 18.1, 19.0586, 20.1, 21.0554, 22.1 \\
\midrule[1pt]
2  & 2.1, 3.0580, 4.1, 5.057 , 6.1, 7.0583, 8.1,   9.0582, 10.1, 11.0580, 12.1,    \\*
   & 13.0582, 14.1, 15.0581, 16.1, 17.0580, 18.1, 19.0585, 20.1, 21.0550, 22.1 \\
\midrule[1pt]
3  & 2.1, 3.0602, 4.1, 5.0578, 6.1, 7.0582, 8.1,   9.0582, 10.1, 11.0580, 12.1,    \\*
   & 13.0582, 14.1, 15.0582, 16.1, 17.0579, 18.1, 19.0585, 20.1, 21.0545, 22.1 \\
\midrule[1pt]
4  & 2.1, 3.0589, 4.1, 5.0571, 6.1, 7.0583, 8.1,   9.0582, 10.1, 11.0580, 12.1,    \\*
   & 13.0582, 14.1, 15.0580, 16.1, 17.0581, 18.1, 19.0580, 20.1, 21.0550, 22.1 \\
\midrule[1pt]
5  & 2.1, 3.0584, 4.1, 5.058 , 6.1, 7.0582, 8.1,   9.0582, 10.1, 11.0580, 12.1,    \\*
   & 13.0582, 14.1, 15.0581, 16.1, 17.0581, 18.1, 19.0582, 20.1, 21.0545, 22.1 \\
\midrule[1pt]
6  & 2.1, 3.0606, 4.1, 5.0575, 6.1, 7.0583, 8.1,   9.0581, 10.1, 11.0580, 12.1,    \\*
   & 13.0582, 14.1, 15.0581, 16.1, 17.0580, 18.1, 19.0585, 20.1, 21.0532, 22.1 \\
\midrule[1pt]
 7  & 2.1, 3.0610, 4.1, 5.0575, 6.1, 7.0583, 8.1,   9.0581, 10.1, 11.0580, 12.1,    \\*
    & 13.0582, 14.1, 15.0581, 16.1, 17.0580, 18.1, 19.0585, 20.1, 21.0532, 22.1 \\
\midrule[1pt]
 8  & 2.1, 3.0589, 4.1, 5.0570, 6.1, 7.0584, 8.1,   9.0581, 10.1, 11.0580, 12.1,    \\*
    & 13.0582, 14.1, 15.0581, 16.1, 17.0580, 18.1, 19.0585, 20.1, 21.0537, 22.1 \\
\midrule[1pt]
 9  & 2.1, 3.0614, 4.1, 5.0577, 6.1, 7.0583, 8.1,   9.0581, 10.1, 11.0580, 12.1,    \\*
    & 13.0582, 14.1, 15.0581, 16.1, 17.0580, 18.1, 19.0582, 20.1, 21.0537, 22.1 \\
\midrule[1pt]
 10 & 2.1, 3.0645, 4.1, 5.0573, 6.1, 7.0582, 8.1,   9.0581, 10.1, 11.0580, 12.1,    \\*
    & 13.0582, 14.1, 15.0581, 16.1, 17.0580, 18.1, 19.0588, 20.1, 21.0541, 22.1 \\
\midrule[1pt]
 11 & 2.1, 3.0606, 4.1, 5.0574, 6.1, 7.0583, 8.1,   9.0582, 10.1, 11.0580, 12.1,    \\*
    & 13.0582, 14.1, 15.0581, 16.1, 17.0580, 18.1, 19.0582, 20.1, 21.0528, 22.1 \\
\midrule[1pt]
 12 & 2.1, 3.0606, 4.1, 5.0574, 6.1, 7.0583, 8.1,   9.0582, 10.1, 11.0580, 12.1,    \\*
    & 13.0582, 14.1, 15.0582, 16.1, 17.0579, 18.1, 19.0583, 20.1, 21.0545, 22.1 \\
\midrule[1pt]
 13 & 2.1, 3.0584, 4.1, 5.0567, 6.1, 7.0580, 8.1,   9.0583, 10.1, 11.0580, 12.1,    \\*
    & 13.0581, 14.1, 15.0581, 16.1, 17.0579, 18.1, 19.0582, 20.1, 21.0541, 22.1 \\
\midrule[1pt]
 14 & 2.1, 3.0580, 4.1, 5.0581, 6.1, 7.0583, 8.1,   9.0581, 10.1, 11.0580, 12.1,    \\*
    & 13.0582, 14.1, 15.0581, 16.1, 17.0579, 18.1, 19.0582, 20.1, 21.0537, 22.1 \\
\midrule[1pt]
 15 & 2.1, 3.0658, 4.1, 5.0578, 6.1, 7.0582, 8.1,   9.0582, 10.1, 11.0580, 12.1,    \\*
    & 13.0582, 14.1, 15.0581, 16.1, 17.0578, 18.1, 19.0583, 20.1, 21.0537, 22.1 \\
\midrule[1pt]
 16 & 2.1, 3.0641, 4.1, 5.0580, 6.1, 7.0580, 8.1,   9.0582, 10.1, 11.0580, 12.1,    \\*
    & 13.0582, 14.1, 15.0581, 16.1, 17.0579, 18.1, 19.0583, 20.1, 21.0537, 22.1 \\
\midrule[1pt]
 17 & 2.1, 3.0628, 4.1, 5.0573, 6.1, 7.0582, 8.1,   9.0582, 10.1, 11.0580, 12.1,    \\*
    & 13.0582, 14.1, 15.0581, 16.1, 17.0581, 18.1, 19.0579, 20.1, 21.0550, 22.1 \\
\midrule[1pt]
 18 & 2.1, 3.0688, 4.1, 5.0579, 6.1, 7.0582, 8.1,   9.0582, 10.1, 11.0550, 12.1,    \\*
    & 13.0582, 14.1, 15.0581, 16.1, 17.0579, 18.1, 19.0583, 20.1, 21.0541, 22.1 \\
\midrule[1pt]
 19 & 2.1, 3.0567 , 4.1, 5.0571, 6.1, 7.0589, 8.1,  9.0582, 10.1, 11.0580, 12.1,    \\*
    & 13.0582, 14.1, 15.0581, 16.1, 17.0580, 18.1, 19.0582, 20.1, 21.0541, 22.1 \\
\midrule[1pt]
 20 & 2.1, 3.0571, 4.1, 5.0573, 6.1, 7.0583, 8.1,   9.0581, 10.1, 11.0580, 12.1,    \\*
    & 13.0582, 14.1, 15.0581, 16.1, 17.0580, 18.1, 19.0582, 20.1, 21.0550, 22.1 \\
\midrule[1pt]
 21 & 2.1, 3.0649, 4.1, 5.0575, 6.1, 7.0583, 8.1,   9.0581, 10.1, 11.0580, 12.1,    \\*
    & 13.0582, 14.1, 15.0581, 16.1, 17.0579, 18.1, 19.0584, 20.1, 21.0541, 22.1 \\
\midrule[1pt]
 22 & 2.1, 3.0628, 4.1, 5.0577, 6.1, 7.0583, 8.1,   9.0581, 10.1, 11.0580, 12.1,    \\*
    & 13.0582, 14.1, 15.0582, 16.1, 17.0579, 18.1, 19.0582, 20.1, 21.0545, 22.1 \\
\midrule[1pt]
 23 & 2.1, 3.0602, 4.1, 5.0572, 6.1, 7.0582, 8.1,   9.0581, 10.1, 11.0581, 12.1,    \\*
    & 13.0582, 14.1, 15.0581, 16.1, 17.0579, 18.1, 19.0587, 20.1, 21.0536, 22.1 \\
\bottomrule[1pt]
\end{longtable}

\bibliographystyle{asa}
\bibliography{seas}

\end{document}